\let\accentvec\vec
\documentclass[%envcountsame,
]{llncs}
\let\vec\accentvec
\usepackage{amsmath}
\usepackage{amssymb}
\usepackage{latexsym}
\usepackage{amsfonts} 
\usepackage{proof}
\usepackage{color}
\usepackage[english]{babel}
\usepackage{hyperref}
\usepackage{stmaryrd}
\usepackage{dice2014}
\usepackage{editing}

\title{A New Term Rewriting Characterisation of
ETIME functions}

\author{Martin Avanzini %\inst{1} 
and
Naohi Eguchi%\inst{2}%
\thanks{The first author is supported by the FWF (Austrian Science Fund) project I-608-N18.
The second author is supported by JSPS posdoctoral fellowships for young scientists.}
}

\institute{Institute of Computer Science, University of Innsbruck, Austria 
\\
\email{\{martin.avanzini,naohi.eguchi\}@uibk.ac.at}
}
\begin{document}

\maketitle

\begin{abstract}
Adopting former term rewriting characterisations of polytime
and exponential-time computable functions,
we introduce a new reduction order, 
the {\em Path Order for ETIME} (\POESTAR\ for short), 
that is sound and complete for ETIME computable functions.
The proposed reduction order for ETIME makes contrasts to those related complexity classes clear.
\end{abstract}

\section{Introduction}

Function-algebraic approaches to computational complexity classes
without explicit bounding constraints have been developed, providing
successful characterisations of
various complexity classes of functions as the smallest classes containing
certain initial functions closed under specific operations.
Initially, S. Bellantoni and S. Cook introduced a
restrictive form of primitive recursion known as {\em safe recursion} \cite{BC92}, or
independently D. Leivant introduced {\em tiered recursion} \cite{Leivant95},
characterising the class of polynomial-time computable functions.
The idea of safe recursion is to separate the arguments of every function into two kinds (by semicolon) so that the number of recursive
calls is measured only by an argument occurring left to semicolon whereas
recursion terms are substituted only for arguments
occurring right:
%
%MA: dyadic notation
\begin{equation}
\tag{\textbf{Safe Recursion}}
  \begin{array}{rcl}
  f(\sn{\varepsilon, \vec y}{\vec z}) &=& g(\sn{\vec y}{\vec z}) \\
  f(\sn{x\cdot i, \vec y}{\vec z}) &=& 
  h(\sn{x, \vec y}{\vec z, f(\sn{x, \vec y}{\vec z})}) \qquad (i = 0,1)
  % f(\sn{0, \vec y}{\vec z}) &=& g(\sn{\vec y}{\vec z}) \\
  % f(\sn{x+1, \vec y}{\vec z}) &=& 
  % h(\sn{x, \vec y}{\vec z, f(\sn{x, \vec y}{\vec z})})
  \end{array}
\end{equation}
In contrast to classical approaches based on bounded recursion, the
function-algebraic characterisation by safe recursion enables us to
define every polytime function by a purely equational system, or in other
word by a term rewrite system.
Improving the function-algebraic characterisation by S.~Bellantoni and S.~Cook, 
together with G.~Moser the authors introduced the 
\emph{(small) polynomial path order} (\POPSTARS)~\cite{AEM12} that 
constitutes an \emph{order-theoretic} characterisation of the polytime functions. 
%MA: dropped ref POP*
% the first author and G. Moser introduced an order, the {\em polynomial path order} (\POPSTAR) \cite{AM08}, and afterwards
% a variant \POPSTARS\ of \POPSTAR\ together with the second author
% \cite{AEM12}, both of which characterise the class of polytime functions.
% MA: notes => work, to avoid confusion
In the present work, we introduce a syntactic extension of \POPSTARS, the {\em Path Order for ETIME} 
(\POESTAR\ for short). 
%MA: added def of ETIME
This order characterises the class of \emph{ETIME} computable functions, i.e., 
functions computable in deterministic time $2^{\bigO(n)}$.

\section{Function-algebraic Backgrounds}
Various function-algebraic characterisations of the ETIME functions are known,
e.g.~\cite{Monien77,clote97}.
It is also known that extension of safe recursion to (multiple) nested
recursion, called {\em safe nested recursion}, captures the class of exponential-time computable 
functions~\cite{AE09}.
Improving the function-algebraic characterisation by safe nested recursion, the authors together with G.~Moser have
introduced an order, the \emph{Exponential Path Order}~(\EPOSTAR), that is sound and
complete for the exponential-time functions. 
%MA: sentence added
The order proposed here is a syntactic restriction of $\EPOSTAR$.
%MA: related work below
% We consider of ETIME functions without constraints on the growth rate,
% i.e., possibly exponential functions computable by 
% deterministic Turing machines in steps bounded by $2^{\bigO(|m|)}$ for the
% binary length $|m| = \lceil \log_2 (m+1) \rceil$ of the maximal
% $m$ among inputs.

It turns out that the following form of safe nested recursion with single recursion arguments
is sound for ETIME functions.
\begin{equation*}
%MA: dyadic notation; added nesting of f
%\tag{\textbf{Safe Recursion}}
  \begin{array}{rcl}
  f(\sn{\varepsilon, \vec y}{\vec z}) &=& g(\sn{\vec y}{\vec z}) \\
  f(\sn{x \cdot i, \vec y}{\vec z}) &=& 
  h(\sn{x, \vec y}%
       {\vec z, f(\sn{x,\vec y}{\vec{h'} (\sn{x, \vec y}{\vec z, f(\sn{x, \vec y}{\vec z})})})})  
       \qquad (i = 0,1)
  % f(\sn{0, \vec y}{\vec z}) &=& g(\sn{\vec y}{\vec z}) \\
  % f(\sn{x+1, \vec y}{\vec z}) &=& 
  % h(\sn{x, \vec y}%
  %      {\vec{h'} (\sn{x, \vec y}%
  %                  {\vec z, f(\sn{x, \vec y}{\vec z})})
  %      })
  \end{array}
\end{equation*}
%
%MA added
The definition of \POESTAR\ essentially encodes this recursion scheme.  
In contrast to related work, this scheme does neither rely on \emph{bounded functions} 
\cite{Monien77} and allows the definition of functions that grow faster than 
a linear polynomial~\cite{clote97}.

%MA: explained above
% The proposed path order \POESTAR\ for ETIME is defined based on this observation as a strictly
% intermediate order that lies between \POPSTARS\ and \EPOSTAR.

%%% Local Variables: 
%%% mode: latex
%%% TeX-master: "paper"
%%% End: 

%\input{rec}
\section{The Path Order for ETIME (\POESTAR)}
We assume at least nodding acquaintance with the basics of term rewriting~\cite{baader}.
For an order $>$, we denote by $\prodext{>}$ the \emph{product extension} of $>$
defined by $\tup[k]{a} \prodext{>} \tup[k]{b}$ if $a_i = b_i$ or $a_i > b_i$ for all 
$i = 1,\dots,k$, and there exists at least one $j \in \{1,\dots,k\}$ such that $a_j > b_j$ holds. 

We fix a countably infinite set of \emph{variables} $\VS$
and a finite set of \emph{function symbols} $\FS$, the \emph{signature}.
The set of terms formed from $\FS$ and $\VS$ is denoted by $\TERMS$.
The signature $\FS$ contains a distinguished set of 
\emph{constructors} $\CS$, elements of $\Val$ are called \emph{values}. 
Elements of $\FS$ that are not constructors are called \emph{defined symbols}
and collected in $\DS$.
We use always $v$ to denote values, and arbitrary terms are denoted by $l,r$ and $s,t,\dots$, possibly followed by subscripts. 
A \emph{substitution} $\sigma$ is a finite mapping from variables to terms, 
its homomorphic extension to terms is also denoted by $\sigma$ and we write $t\sigma$ 
instead of $\sigma(t)$. 

A \emph{term rewrite system} (\emph{TRS} for short) $\RS$ (over $\FS$) is a finite set of rewrite rules $f(\seq{l}) \to r$, 
where all variable in the term $r$ occur in the term $f(\seq{l})$
and $f \in \DS$.
Adopting \emph{call-by-value} semantics, we define the 
\emph{rewrite relation} $\rew[\RS]$ by
\begin{equation*}
  (\text{i})~\frac{f(l_1,\dots,l_n) \to r \in \RS,~\ofdom{\sigma}{\VS \to \Val}}
       {f(l_1\sigma,\dots,l_n\sigma) \rew r\sigma}
\quad
  (\text{ii})~\frac{s \rew t}
       {f(\dots,s,\dots) \rew f(\dots,t,\dots)}  \tpkt
\end{equation*}

Throughout the present notes we only consider \emph{completely defined},%
\footnote{The restriction is not necessary, but simplifies our presentation.}
\emph{orthogonal constructor}~TRSs~\cite{baader}, that is, 
for each application of $(i)$ there is exactly one matching rule $f(l_1,\dots,l_n) \to r \in \RS$; 
the terms $l_i$ ($i = 1,\dots,n$) contains no defined symbols and 
variables occur only once in $f(l_1,\dots,l_n)$. 

For each defined symbol $f$ of arity $k$, $\RS$ defines a
function $\ofdom{\sem{f}}{\Val^k \to \Val}$ by
$\sem{f}(\seq[k]{v}) \defsym v$ iff $f(\seq[k]{v}) \rew[\RS] \cdots \rew[\RS] v$. 
These functions are well-defined if $\RS$ \emph{terminating}, i.e., when $\rew[\RS]$ is well-founded.
We do not presuppose that $\RS$ is terminating, instead, our method implies termination. 

For a term $t$, the \emph{size} of $t$ is denoted as $\size{t}$ referring to the number of symbols occurring in $t$.  
%the \emph{depth} $\depth(t)$ of $t$ is given by the length of the longest path in $t$ when conceived as a tree.
For a complexity measure for TRSs, the \emph{(innermost) runtime complexity function} $\ofdom{\rc[\RS]}{\N \to \N}$ is defined by
$$
 \rc[\RS](n) \defsym 
 \max\{\ell \mid \exists s = f(\seq{v}), \size{s} \leqslant n \text{ and } s  = t_0 \rew t_1 \rew \dots \rew t_\ell\} 
 \tkom
$$
which is well-defined for terminating TRSs $\RS$. 
The runtime-complexity function constitutes an \emph{invariant cost-model} for rewrite systems:
the functions $\sem{f}$ ($f \in \DS$) can be computed within polynomial overhead
on conventional models of computation, e.g., on Turing machines~\cite{LM09,AM10}.

Let $\sp$ denote a strict order on $\FS$, the \emph{precedence}. 
We assume that the argument positions of every function symbol
are separated into two kinds.
The separation is denoted by semicolon as
$f(\sn{t_1, \dots, t_k}{t_{k+1}, \dots, t_{k+l}})$,
where $t_1, \dots, t_k$ are called {\em normal} arguments whereas
$t_{k+1}, \dots, t_{k+l}$ are called {\em safe} ones.
For constructors $\CS$, we suppose that all symbols are safe. 
We write $s \nsupertermstrict t$ if $t$ is a \emph{sub-term of 
a normal argument} of $s$, i.e., $s = f(\pseq[k][l]{s})$ and 
$t$ occurs in a term $s_i$ for $i \in \{1,\dots,k\}$. 
The following definition introduces the 
instance $\poe$ of the $\POESTAR$ as induced by a precedence $\sp$. 
\begin{definition}
\label{d:poe}
  Let $\sp$ be a precedence and let $s = f(\sn{s_1, \dots, s_k}{s_{k+1}, \dots, s_{k+l}})$.
  Then $s \poe t$ 
  if one of the following alternatives holds.
  \begin{enumerate}
  \item \label{d:poe:st}
    $s_i \eqpoe t$ for some 
    $i \in \{1, \dots, k+l \}$.
  \item \label{d:poe:ia}
    $f \in \DS$ and 
    $t = g(\sn{t_1, \dots, t_m}{t_{m+1}, \dots, t_{m+n}})$ 
    % for some $g \in \FS$ and for some
    % $t_1, \dots, t_{m+n} \in \TERMS$ 
    % such that
    with
    $f \sp g$ and:
    \begin{itemize}
    \item $s \nsupertermstrict t_j$ for all  
      $j \in \{ 1, \dots, m \}$;
      \label{d:poe:ia:1}
    \item $s \poe t_j$ for all $j \in \{ m+1, \dots, m+n \}$;
    \label{d:poe:ia:2}
    \end{itemize}
  \item $f \in \DS$ and 
        $t = f(\sn{t_1, \dots, t_k}{t_{k+1}, \dots, t_{k+l}})$ 
        % for some $g \in \FS$ and for some
        % $t_1, \dots, t_{k+l} \in \TERMS$ such that 
        and:
  \label{d:poe:ts}
    \begin{itemize}
    \item $\tuple{s_1,\dots,s_k} \prodext{\poe} \tuple{t_{1},\dots,t_{k}}$ 
    \item $s \poe t_j$ for all $j \in \{ k+1, \dots, k+l \}$. 
    \end{itemize}
  \end{enumerate}
\end{definition}

We say that $\RS$ is \emph{\poecompatible} if 
for some precedence $\sp$, 
$l \poe r$ holds for all rules $l \to r \in \RS$. 
%\begin{todos}
%  \item examples
%\end{todos}

\begin{example}
The standard addition $(x, y) \mapsto x+y$ (in unary notation) is defined by a TRS
$\RS_{\m{add}}$ consisting of the following two rules.
\label{ex:add}
  \begin{alignat*}{4}
    \rlbl{1}: &~& \m{add}(\sn{\mZ}{y}) & \to y  \qquad\qquad &
    \rlbl{2}: &~& \m{add}(\sn{\ms(\sn{}{x})}{y}) & 
    \to \ms (\sn{}{\m{add} (\sn{x}y)})
  \end{alignat*}
Define a precedence by $\m{add} \sp \ms$ and an argument separation as
indicated in the rules.
Then it can be seen that 
$\m{add}(\sn{\mZ}{y}) \poe y$
and
$\m{add}(\sn{\ms(\sn{}{x})}{y}) \poe \ms (\sn{}{\m{add} (\sn{x}y)})$
hold for the order $\poe$ induced by the precedence $\sp$.
\end{example}

\begin{example}
An exponential $2^{x} + y$ is defined by a TRS $\RS_{\m{exp}}$ consisting of the following two rules. 
\label{ex:exp}
  \begin{alignat*}{4}
    \rlbl{1}: &~& \m{exp}(\sn{\mZ}{y}) & 
    \to \ms (\sn{}{y})  \qquad\qquad &
    \rlbl{2}: &~& \m{exp}(\sn{\ms(\sn{}{x})}{y}) & 
    \to \m{exp} (\sn{x}{\m{exp} (\sn{x}{y})})
  \end{alignat*}
The TRS $\RS_{\m{exp}}$ is compatible with the order $\poe$ induced
by the precedence $\m{exp} \sp \ms$.
\end{example}

\begin{example}
A factorial function of the form $y \cdot x ! + z$ is defined by a 
TRS $\RS_{\m{fac}}$ consisting of $\RS_{\m{add}}$ and additionally of
 the following three rules.
\label{ex:fac}
  \begin{alignat*}{4}
    \rlbl{3}: &~& \m{fac} (\sn{\mZ, y}{z}) & 
    \to \m{add} (\sn{y}{z}) &
    \rlbl{4}: &~& \m{fac} (\sn{\ms(\sn{}{x}), \mZ}{z}) &  
    \to z
    \\
    \rlbl{5}: &~& \m{fac}(\sn{\ms (\sn{}{x}), \ms (\sn{}{y})}{z}) & 
    \to \m{fac} (\sn{\ms (\sn{}{x}), y}%
                    {\m{fac} (\sn{x, \ms (\sn{}{x})}{z})
                                 })
 & &
  \end{alignat*}
The TRS $\RS_{\m{fac}}$ is not compatible with any \POESTAR.
In particular, rule 5 is not orientable since element-wise
 comparison of
$\tuple{\ms (\sn{}{x}), \ms (\sn{}{y})}$ 
and
$\tuple{x, \ms (\sn{}{x})}$
fails.
\end{example}
Note that function $\sem{\m{fac}}$ is computable in exponential-time, 
but not in ETIME. 
%MA: added. 
% By the following Theorem it is not surprising that $\RS_{\m{fac}}$ is not 
% \poecompatible. 

\begin{theorem}[Soundness of \POESTAR for ETIME]
\label{t:sound}
Every function defined by a \poecompatible\ rewrite system is ETIME computable.
\end{theorem}
This theorem follows from the following key lemma, whose proof is involved and hence postponed to the next section.

\begin{lemma}\label{l:poe}
For any \poecompatible\ rewrite system $\RS$, $\rc[\RS](n) \in 2^{\bigO(n)}$ holds. 
\end{lemma}

%\begin{theorem}\label{t:E}
%  The following class of functions are equivalent:
%  \begin{enumerate}
%  \item\label{t:E:1} The class of functions computed \poecompatible\ TRSs.
%  \item\label{t:E:3} The class ETIME of functions computable in time $2^{\bigO(n)}$. 
%  \end{enumerate}
%\end{theorem}

%MA: moved after soundness
Although the inverse of Lemma~\ref{l:poe} is in general not true, 
the order is also extensionally complete for the ETIME functions. 
\begin{theorem}[Completeness of \POESTAR\ for ETIME]
\label{t:comp}
Every ETIME function can be defined by a \poecompatible\ rewrite system.
\end{theorem}
\begin{proof}[Sketch]
  Consider words formed from dyadic successors $\m{0}$ and $\m{1}$ together with 
  a constant $\m{\epsilon}$, denoting the empty word.
  The following rewrite rules
  \[
  \m{f}_1 (\sn{\m{\epsilon}}{u}) \rightarrow \m{d}(\sn{}{u}) \qquad\quad
  \m{f}_1 (\sn{ \m{i}(\sn{}{x})}{u}) \rightarrow
  \m{f}_1 (\sn{x}{\m{f}_0(\sn{x}{u})}) 
  \quad (\text{for}~\m{i} \in \{\m{0},\m{1}\})
  \tkom
  \]
  define a function $\sem{\m{f}_1}(\sn{w}{c}) = \sem{\m{d}}^{2^{\size{w}}}(\sn{}{c})$, 
  i.e., $2^{\size{w}}$-fold iteration of $\sem{\m{d}}$. Here, we suppose 
  that $\size{w}$ counts the number of occurrences of $\m{0}$ and $\m{1}$ 
  in $w$. 
  Next consider the following rewrite rules.
  \[
  \m{f}_2 (\sn{\m{\epsilon}, y}{u}) \rightarrow 
  \m{f}_1 (\sn{y}{u}) \qquad
  \m{f}_2 (\sn{\m{i}(\sn{}{x}), y}{u}) \rightarrow
  \m{f}_2 (\sn{x,y}{\m{f}_2 (\sn{x,y}{u})})
  \quad (\text{for}~\m{i} \in \{\m{0},\m{1}\})
  \tpkt
  \]
  Then $\sem{\m{f}_2}(\sn{w,w}{u}) = \sem{\m{d}}^{2^{2 \cdot \size{w}}} (\sn{}{u})$.
  This construction can be extended to $k$ functions
  such that
  $\sem{\m{f}_k}(\sn{w,\dots,w}{u}) = \sem{\m{d}}^{2^{k \cdot \size{w}}} (\sn{}{u})$. 
  Note that all rules can be oriented by $\poe$, given 
  by the precedence $\m{f}_k \sp \cdots \sp \m{f}_1 \sp \m{d}$.

  Using this construction, it is possible to simulate ETIME Turing machine computations 
  by a \poecompatible\ TRS, essentially by substituting the transition 
  function for $\m{d}$. Note that $\m{d}$ can be defined by pattern matching only, 
  in particular it is easy to define $\m{d}$ such that the underlying rewrite rules are \poecompatible. 
\end{proof}

\begin{corollary}
The class of ETIME computable functions coincides with the class of
 functions computed by \poecompatible\ rewrite systems.
\end{corollary}

%%% Local Variables: 
%%% mode: latex
%%% TeX-master: "paper"
%%% End: 

\section{Soundness Proof}
\label{s:sound}

In this section we prove Theorem~\ref{t:sound}.
The proof follows the pattern of the proof of 
soundness for the exponential path order~\cite{AEM11}.
We embed reductions of \poecompatible~rewrite systems 
into an auxiliary order $\poel$, whose length of 
maximal descending sequences we estimate appropriately
below. 

\subsection{Order on Sequences}
To formalise sequences of terms, 
we use an auxiliary variadic function symbol $\listsym$. 
Here variadic means that the arity of $\listsym$ 
is finite but arbitrary.  
Call a term $t \in \TA(\FS \cup \{\listsym\}, \VS)$ a \emph{sequence} if it 
is of the form $\listsym(\seq[k]{t})$ for $t_i \in \TERMS$ ($i = 1,\dots,k$). 
We always write $\lseq[k]{t}$ instead of $\listsym(\seq[k]{t})$.
We use $a,b,\dots$ to denote terms and sequences of terms.
We define \emph{concatenation} as $\lseq[k]{s} \append \lseq[l]{t} \defsym \lst{s_1~\cdots~s_k~t_1~\cdots~t_l}$, and 
extend it to terms by identifying terms $t$
with the singleton sequences $\lst{t}$, for instance $s \append t= \lst{s~t}$. 
For sequences $a$ define $\tolst(a) \defsym a$, and for terms $t$ define $\tolst(t) \defsym \lst{t}$.
\begin{definition}\label{d:poel} 
  Let $\spl$ denote a precedence on $\FS$.
  Let $\ell \in \N$ with $\ell \geqslant 1$.
  Then $a \poel b$ holds for terms or sequences of terms 
  $a,b$ if one of the following alternatives hold.
  \begin{enumerate}
  \item\label{d:poel:ia} 
    $a = f(\seq[k]{s})$, $b=g(\seq[l]{t})$ with $f \spl g$ and the following conditions hold:
    \begin{itemize}
    \item $f(\seq[k]{s}) \supertermstrict t_j$ for all $j = 1,\dots,l$; and
    \item $l \leqslant \ell$; or
    \end{itemize}
  \item\label{d:poel:ts}
    $a = f(\seq[k]{s})$, $b=f(\seq[k]{t})$ and $\tup[k]{s} \prodext{\supertermstrict} \tup[k]{t}$;
  \item\label{d:poel:ialst} 
    $a = f(\seq[k]{s})$, $b=\lseq[l]{t}$ and the following conditions hold:
    \begin{itemize}  
    \item $f(\seq[k]{s}) \poel t_{j}$ for all $j = 1,\dots,l$; and
    \item $l \leqslant \ell$.
    \end{itemize}
  \item\label{d:poel:ms} 
    $a = \lseq[k]{s}$, $b=\lseq[l]{t}$ and there exists terms \emph{or} sequences $b_i$ ($i = 1,\dots,k$) such that:
    \begin{itemize}
    \item $\lseq[l]{t} = b_1 \append \cdots \append b_k$; and
    \item $\tup[k]{s} \prodext{\poel} \tup[k]{b}$.
    \end{itemize}
  \end{enumerate}
\end{definition}

For notational convention, we will write $s \caseref{\poel}{i} t$ if 
$s \poel t$ follows from the \nth{$i$} clause in Definition~\ref{d:poel}. 
The following lemma collects frequently used properties of $\poel$.
\begin{lemma}\label{l:poel:approx}
  Let $\ell \geqslant 1$.
  The order $\poel[\ell]$ satisfies the following properties:
  \begin{enumerate}
  \item\label{l:poel:approx:inc} ${\poel} \subseteq {\poel[\ell+1]}$; and
  \item\label{l:poel:approx:seq} 
    if $a \poel b$ then ${c_1 \append a \append c_2} \poel {c_1 \append b \append c_2}$
    for all terms or sequences $a,b,c_1,c_2$.
  \end{enumerate}
\end{lemma}
\begin{proof}
  Properties~\eqref{l:poel:approx:inc} follows by definition.
  To prove the second property, 
  suppose $a \poel b$ holds. Set $\lseq[k]{u} \defsym \tolst(c_1)$ and $\lseq[l]{v} \defsym \tolst(c_2)$, 
  and observe that by the overloading of $\append$ we have 
  \begin{align*}
    c_1 = u_1 \append \cdots \append u_k \quad \text{and} \quad c_2 = v_1 \append \cdots \append v_l \tpkt
  \end{align*}
  If $a = f(\seq[m]{s})$ is a term, 
  % Then 
  % \begin{align*}
  %   c_1 \append a \append c_2 & = u_1 \append \cdots \append u_k \append f(\seq[m]{s}) \append v_1 \append \cdots \append v_l \text{, and} \\
  %   c_1 \append b \append c_2 & = u_1 \append \cdots \append u_k \append b \append v_1 \append \cdots \append v_l
  % \end{align*}
  then by assumption $f(\seq[m]{s}) \poel b$ we have 
  $$
  \tuple{\seq[k]{u},f(\seq[m]{s}),\seq[l]{v}}
  \poel
  \tuple{\seq[k]{u},b,\seq[l]{v}}
  \tkom
  $$
  and thus
  $$
   c_1 \append f(\seq[m]{s}) \append c_2 \cpoel{ms} c_1 \append b \append c_2 \tkom
  $$
  holds as desired.
  Otherwise $a = \lseq[m]{s}$, 
  and the assumption can be strengthened to $a \cpoel{ms} b$. 
  By definition $b = b_1 \append \cdots \append b_m$ for some terms or sequences $b_j$ ($j = 1,\dots,m$) with 
  $\tup[m]{s} \poel \tup[m]{b}$.
  From this we obtain 
  $$
  \tuple{\seq[k]{u},\seq[m]{s},\seq[l]{v}}
  \poel
  \tuple{\seq[k]{u},\seq[m]{b},\seq[l]{v}}
  \tpkt
  $$
  Hence again the property follows by one application of $\cpoel{ms}$. 
  \qed
\end{proof}

Note that the order $\poel[\ell]$ is a restriction of the \emph{multiset path order}~\cite{baader} \scmt{really}
using $\listsym$ as a minimal element. The order is thus well-founded.
Since the indices $\ell$ ensures that $\poel[\ell]$ is finitely branching
the length of the maximal $\poel[\ell]$-descending sequence, expressed by the function $\Slow[\ell]$
is well-defined.

\begin{definition}\label{d:poel:slow}
  For $\ell \geqslant 1$, and terms or sequences $a$ define
  \[
  \Slow[\ell](a) \defsym \max \{~l~\mid \exists a_1,\dots,a_l.~ a \poel a_1 \poel \cdots \poel a_l \}
  \tpkt
  \]
\end{definition}

Note that if $a \poel b$ then $\Slow(a) > \Slow(b)$ holds. 
In the following, we prove that $\Slow(a)$ 
is bounded by an exponential in the depth of its argument. 
The following lemma serves as an auxiliary step.

\begin{lemma}\label{l:slowsum} 
  For all $\ell \geqslant 1$ and sequences $\lseq[k]{t}$, we have
  $\Slow(\lseq[k]{t}) = \sum_{i=1}^k \Slow(t_i)$. 
\end{lemma}
\begin{proof}
  Consider a sequence $\lseq[k]{t}$. 
  As a consequence of Lemma~\eref{l:poel:approx}{seq},
  $\Slow(a \append b) \geqslant \Slow(a) + \Slow(b)$, 
  holds for all sequences and terms $a,b$. 
  Hence in particular
  $\Slow(\lseq[k]{t}) =  \Slow(t_1 \append \cdots \append t_k) \geqslant \sum_{i=1}^k \Slow(t_i)$. 

  To show the inverse direction, we proceed by induction on $\Slow(\lseq[k]{t})$.
  The base case $\Slow(a) = 0$ follows trivially.
  For the induction step, we show that for all terms or sequences$b$, 
  $\lseq[k]{t} \poel b \IImp \Slow(b) < \sum_{i=1}^k \Slow(t_i)$
  holds.
  This implies
  $\Slow(\lseq[k]{t}) \leqslant \sum_{i=1}^k \Slow(t_i)$ as desired.
  Suppose $a \poel b$, which by definition of $\poel$ refines to $a \cpoel{ms} b$.
  Hence there exists terms or sequences $b_i$ ($i = 1,\dots,k$)
  such that $b = b_1 \append \cdots \append b_k$ and 
  $$
  \tup[k]{t} \poel \tup[k]{b} \tkom
  $$
  holds. 
  As a consequence,
  $\Slow(b_i) \leqslant \Slow(t_i)$ holds for all $i = 1,\dots,k$, 
  where for at least one $i_0 \in \{1,\dots,k\}$ we even have
  $\Slow(b_{i_0}) < \Slow(s_{i_0})$. 
  Using that  
  $$
  \Slow(b_i) \leqslant \Slow(b) < \Slow(\lseq[k]{t}) \quad \text{for all $i =1,\dots,k$,}
  $$
  induction hypothesis is applicable to $b$ and all $b_i$ ($i \in \{1,\dots,k\}$).
  Summing up we obtain
  \begin{align*}
    \Slow(b) = \sum_{s \in b} \Slow(s) = \sum_{i=1}^k \sum_{s \in b_i} \Slow(s) 
    = \sum_{i=1}^k \Slow(b_i) < \sum_{i=1}^k \Slow(t_i) \tpkt
  \end{align*}
  \qed
\end{proof}

\begin{lemma}\label{l:poel}
  Given a precedence $\spl$ on a signature $\FS$, 
  the {\em rank} $\rk: \FS \rightarrow \mathbb N$ is defined in
 accordance with $\spl$ as
  $\rk (f) > \rk (g) \Leftrightarrow f \spl g$.
  Let $\ell \geqslant 1$. 
  Then, for any function symbol $f \in \FS$ with arity 
  $n \leqslant \ell$ and for any 
  terms $\seq[n]{t}$, the following inequality holds.
  \begin{equation}
  \label{eq:poel:l}
  \Slow(f(\seq[n]{t})) \leqslant 
  (\ell +1)^{(\ell +1)^{\rk (f)} \cdot 
             \bigl( \textstyle\sum_{i=1}^n \Slow (t_i) \bigr)       
            }.
  \end{equation}
\end{lemma}

\begin{proof}
Let $t = f(t_1, \dots, t_n)$.
We prove the inequality (\ref{eq:poel:l})
by induction on 
$\Slow (t)$.
In the base case, $\Slow (t) =0$, and hence the inequality
 (\ref{eq:poel:l}) trivially holds.
In the case $\Slow (t) > 0$, it suffices to show that for any
$b \in  \TA(\FS \cup \{\listsym\}, \VS)$,
$t \poel b$ implies
$\Slow (b) <
  (\ell +1)^{(\ell +1)^{\rk (f)} \cdot 
             \bigl( \textstyle\sum_{i=1}^n \Slow (t_i) \bigr)       
            }
$.
The induction case splits into three cases 
$t \caseref{\poel}{i} b$ 
($i \in \{ \ref{d:poel:ia}, \ref{d:poel:ts}, \ref{d:poel:ialst}\}$). 
%(Definition \ref{d:epo}.\ref{d:epo:4} is not the case.)
For the sake of convenience, we start with the case 
$t \cpoel{ialst} b$.
Namely, we consider the case 
$b = \lseq[k]{s}$ where $1 \leqslant k \leqslant \ell$
    and $t \poel s_i$ for all $i \in \{1, \dots, k \}$.
We show that for all $i \in \{ 1, \dots, k \}$, 
\begin{equation}
 \Slow (s_i) \leqslant 
  (\ell +1)^{(\ell +1)^{\rk (f)} \cdot 
             \bigl( \textstyle\sum_{i=1}^k \Slow (t_i) 
             \bigr) -1       
            }
 \tpkt
\label{eq2:t:poel}
\end{equation}
We prove the inequality (\ref{eq2:t:poel}) by case analysis depending on
$j \in \{ \ref{d:poel:ia}, \ref{d:poel:ts} \}$ where 
$t \caseref{\poel}{j} s_i$ holds.
Fix some element 
$u \in \{ s_i \mid i\in \{1, \dots, k\} \}$.
\begin{enumerate}
\item \textsc{Case}. $t \cpoel{ia} u$: $u = g(\seq[m]{u})$ where $m \leqslant \ell$, 
  $g$ is a defined symbol with $f \spl g$ 
  and for all $i \in \{ 1, \dots, m\}$, $t$ is a strict super-term of $u_i$.
  We reason
  \begin{align*} 
  &&& 
  (\ell +1)^{\rk (g)} \cdot 
  \bigl( \textstyle\sum_{i=1}^m \Slow (u_i) \bigr) \\
  && \leqslant &
  (\ell +1)^{\rk (g)} \cdot 
  \bigl( m \cdot \max_{1 \leqslant i \leqslant n} \Slow (t_i) \bigr) \\
  && < &
  (\ell +1)^{\rk (g)} \cdot 
  \bigl( (\ell +1) \cdot \textstyle\sum_{i=1}^n \Slow (t_i) \bigr) 
  & (\text{since } m \leqslant \ell) \\
  && \leqslant &
  (\ell +1)^{\rk (f)} \cdot 
  \bigl( \textstyle\sum_{i=1}^n \Slow (t_i) \bigr).
  & (\text{since } \rk (g) < \rk (f))
  \end{align*}
  This together with induction hypothesis allows us to derive the inequality (\ref{eq2:t:poel}).
\label{case:ia}
\item \textsc{Case}. $t \cpoel{ts} u$:  $u = f(\seq[n]{u})$ where
  $\tup[n]{t} \prodext{\supertermstrict} \tup[n]{u}$
  holds.
  In this case the inequality (\ref{eq2:t:poel}) follows from induction hypothesis together together an easy
      observation that
  $\textstyle\sum_{i=1}^n \Slow (u_i) <
   \textstyle\sum_{i=1}^n \Slow (t_i)$.
\label{case:ts}
\end{enumerate}
Summing up Case \ref{case:ia} and \ref{case:ts} concludes inequality (\ref{eq2:t:poel}).
Thus, having $\Slow (b) = \sum_{i=1}^{k} \Slow (s_i)$ by Lemma \ref{l:slowsum}, and employing $k \leqslant \ell$, we see
\begin{align*}
  \Slow (b) %= \sum_{i=1}^{m} \Slow (s_i) 
  & \leqslant 
  \ell \cdot
  (\ell +1)^{(\ell +1)^{\rk (f)} \cdot 
             \bigl( \textstyle\sum_{i=1}^k \Slow (t_i) 
             \bigr) -1       
            }
  & \text{(by the inequality (\ref{eq2:t:poel}))} \\
  & < 
  (\ell +1)^{(\ell +1)^{\rk (f)} \cdot 
             \bigl( \textstyle\sum_{i=1}^k \Slow (t_i) \bigr)       
            }
  \tpkt
\end{align*}
This completes the case $t \cpoel{ialst} b$. 
The cases $t \cpoel{ia} b$ and $t \cpoel{ts} b$ follow 
 respectively from 
 Case \ref{case:ia} and \ref{case:ts}. 
\qed
\end{proof}

\subsection{Predicative Embedding of $\rew[\RS]$ into $\poel$}
Throughout the following, we fix a \poecompatible\ TRS $\RS$. 
We now establish the \emph{predicative embedding} of $\rew[\RS]$ into 
the order $\poel[\ell]$, for $\ell$ depending only on $\RS$. 
The \emph{predicative interpretation $\pint$} that we use in this embedding
separates safe from normal arguments resulting in a sequences of \emph{normalised terms}. 

\begin{definition}
  For each $f \in \FS$ with $k$ normal arguments, let $\fn$ denote a fresh function symbol
  of arity $k$.
  We set $\FSn \defsym \FS \cup \{ \fn \mid f \in \FS \}$. 
  A term $t \in \TA(\FSn,\VS)$ is called \emph{normalised} if it is of the form 
  $t = \fn(\seq[k]{t})$ for $t_i \in \TERMS$ ($i = 1,\dots,k$). 
\end{definition}

\begin{definition}\label{d:pint}
  We define the \emph{predicative interpretation $\pint$}, mapping terms to sequences of normalised terms, as follows:
\begin{equation*}
  \pint(t) \defsym 
  \begin{cases}
    \nil & \text{if $t \in \Val$,} \\
    \fn(\seq[k]{t}) \append \pint(t_{k+1}) \append \cdots \append \pint(t_{k+l}) 
    & \text{if $t \not\in \Val$.}
  \end{cases}
\end{equation*}
For the second case we suppose $t = f(\pseq{t})$.
\end{definition}

In the following, we show that a reduction
$f(\seq[k]{v}) \rew[\RS] s_1 \rew[\RS] s_2 \rew[\RS] \dots$ 
translates
into a sequence $\pint(f(\seq[k]{v})) \poel \pint(s_1) \poel \pint(s_2) \poel \dots$ 
for $\ell$ the maximal size of a right-hand side in $\RS$.
In the embedding, we use as precedence the projection of the precedence $\sp$ 
underlying $\RS$ to the normalised signature $\FSn$, defined by
\begin{align*}
  \fn \spl \gn \defiff f \sp g \quad\text{and}\quad f \spl g \defiff f \sp g \tpkt
\end{align*}

In the proof of this \emph{embedding} it is important 
to notice that rewriting happens never under normal argument positions. 
To this end we introduce a set $\Tn$, consisting of terms where normal arguments are values. 
\begin{definition}
We define $\Tn$ as the least such that 
(i) $\Val \subseteq \Tn$, and
(ii) if $f \in \FS$, $\seq[k]{v} \in \Val$ 
and $\seq[l]{t} \in \Tn$ 
then $f(\sn{\seq[k]{v}}{\seq[l]{t}}) \in \Tn$.
\end{definition}
Observe that $f(\seq[k]{v}) \in \Tn$ for values $v_i$ ($i = 1,\dots,k$). 
The set $\Tn$ is closed under rewriting. 
\begin{lemma}\label{l:tnderiv}
  If $s \in \Tn$ and $s \rew t$ then $t \in \Tn$.
\end{lemma}
\begin{proof}
  The lemma follows a standard induction on the definition of $\poe$. 
\end{proof}

\begin{lemma}\label{l:embed:root}
  Let $\ofdom{\sigma}{\VS \to \Val}$ be a substitution. 
  Then $l\sigma \poel[\size{r}] r\sigma$ for all rule $l \to r \in \RS$.
\end{lemma}
According to  $s \caseref{\poel}{i} t$, we write 
$s \caseref{\poe}{i} t$ if 
$s \poe t$ follows from the \nth{$i$} clause in Definition~\ref{d:poe}. 

\begin{proof}
  Fix terms $s = f(\pseq[k][l]{s})$ with $f \in \DS$
  and $\seq[k+l]{s} \in \TA(\CS,\VS)$. 
  We first show that for all terms $t$, 
  \begin{equation}
    \label{eq:poe:embed}
    \tag{$\dag$}
    s \poe t \IImp \text{$\fn(s_1\sigma, \dots, s_k\sigma) \poel[\size{t}] u$ for all $u \in \pint(t\sigma)$}
    \tpkt
  \end{equation}
  Suppose $s \poe t$ holds, the proof is by induction on $\size{t}$. 
  The non-trivial case is when $t\sigma \not \in \Val$ as otherwise 
  $\pint(t\sigma) = \nil$.
  This excludes a priori the case $s \cpoe{st} t$ by the assumption on the shape of $s$. 
  Hence either  $s \cpoe{ia} t$ or $s \cpoe{ts} t$ holds, 
  and thus $t = g(\pseq[m][n]{t})$ for some $g \in \FS$
  and terms $t_j$ ($j = 1,\dots,m+n$), 
  By definition, 
  $$
  \pint(t\sigma) 
  = \lst{\gn(t_1\sigma, \dots, t_m\sigma)} 
  \append \pint(t_{m+1}\sigma) \append \cdots \append \pint(t_{m+n}\sigma) 
  \tpkt
  $$
  To prove the implication~\eqref{eq:poe:embed}, 
  consider first the element $u = \gn(t_1\sigma, \dots, t_m\sigma) \in \pint(t\sigma)$.

  Suppose first that $s \cpoe{ia} t$ holds. 
  Thus $f \sp g$ and hence $\fn \spl \gn$.
  Consider a normal argument position $j \in \{1,\dots,m\}$ of $g$. 
  The assumption $s \poe t$ gives $s \supertermstrict t_j$.
  Hence there exists a normal argument position $i \in \{1,\dots, k\}$ of $f$ 
  with $s_i \superterm t_j$, and hence $s_i\sigma \superterm t_j\sigma$ holds. 
  In total, $\fn(s_1\sigma,\dots,s_k\sigma) \supertermstrict t_j\sigma$ holds for all $j = 1,\dots,m$.
  Since trivially $m \leqslant \size{t}$, 
  we conclude 
  $\fn(s_1\sigma,\dots,s_k\sigma) \cpoel[\size{t}]{ia}  \gn(t_1\sigma, \dots, t_m\sigma)$ as desired.

  Finally, suppose $s \cpoe{ts} t$ holds, thus $t = f(\pseq[k][l]{t})$.
  Since $\tup[k]{s} \poeprod \tup[k]{t}$ holds in this case 
  and $s_i \in TA(\CS,\VS)$ for all $i = 1,\dots,k$, 
  it is not difficult to conclude that 
  $s_i \superterm t_i$ and hence
  $\tuple{s_1\sigma,\dots,s_k\sigma} \supertermstrict \tuple{t_1\sigma,\dots,t_k\sigma}$, 
  holds. 
  As trivially $k \leqslant \ell$, 
  we $\fn(s_1\sigma,\dots,s_k\sigma) \cpoel[\size{t}]{ts}  \gn(t_1\sigma, \dots, t_k\sigma)$ in this final case.

  Now consider the remaining elements $u \in \pint(t\sigma)$, $u \not=\gn(t_1\sigma, \dots, t_m\sigma)$.
  Then $u$ occurs in the interpretation of a safe argument of $t\sigma$ by definition of 
  the interpretation, say $u \in \pint(t_j\sigma)$
  for some $j \in \{m+1,\dots,m+n\}$. 
  One verifies that $s \poe t_j$ holds: 
  in the case $s \cpoe{ts} t$ we have $s \poe t_j$ by definition; 
  otherwise $s \cpoe{ia} t$ holds and we even obtain $s \cpoe{st} t_j$.
  As $\size{t_j} < \size{t}$, by induction hypothesis we have $\fn(s_1\sigma,\dots,s_k\sigma) \poel[\size{t_j}] u$, 
  and thus $\fn(s_1\sigma,\dots,s_m\sigma) \poel[\size{t}] u$ using Lemma~\eref{l:poel:approx}{inc}.
  Overall, we conclude the implication~\eqref{eq:poe:embed}.

  Fix a rule $l \to r \in \RS$ with $l = f(\seq[m][n]{l})$.
  We return to the main proof, and show that
  $\fn(l_1\sigma,\dots,l_m\sigma) \cpoel[\size{r}]{ialst} \pint(r\sigma)$ holds, 
  from which the lemma follows by one application of \cpoel[\size{r}]{ms}.
  \begin{itemize}
  \item The implication~\eqref{eq:poe:embed} and assumption $l \poe r$ 
    gives $\fn(l_1\sigma,\dots,l_m\sigma) \poel[\size{r}] u$ 
    and thus $\fn(l_1\sigma,\dots,l_m\sigma) \poel[\size{r}] u$ 
    for all elements $u \in \pint(r\sigma)$. 
  \item The length of $\pint(r\sigma)$ is bounded by $\size{r}$. 
    This can be shown by a standard induction on $r$, using 
    in the base case $x \in \VS$ that $\pint(r\sigma) = \nil$ as $\sigma(x) \in \Val$ for all variables
    in $r$. 
  \end{itemize}
\end{proof}

\begin{lemma}\label{l:embed}
  Let $\RS$ be a \poecompatible\ TRS and let $\ofdom{\sigma}{\VS \to \Val}$ be a substitution, 
  and let $\ell$ denote the maximal size of right-hand sides $r$ of rules $l \to r \in \RS$. 
  If $s \in \Tn$ and $s \rew[\RS] t$ then 
  $\pint(s) \poel[\ell] \pint(t)$. 
\end{lemma}
\begin{proof}
  Let $s \in \Tn$ and consider a rewrite step $s \rew[\RS] t$. 
  The base case is covered by Lemma~\ref{l:embed:root}, 
  hence consider a rewrite step below the root. 
  Since $s \in \Tn$, in Lemma~\ref{l:tnderiv} we already observed that this step is of the form
  \begin{multline*}
  s = f(\sn{\seq[m]{s}}{s_{m+1},\dots,s_i,\dots,s_{m+n}}) \\
  \rew[\RS] f(\sn{\seq[m]{s}}{s_{m+1},\dots,t_i,\dots,s_{m+n}}) = t 
  \tkom
  \end{multline*}
  where $s_i \rew[\RS] t_i$.
  The non-trivial case is when $t \not\in \Val$, otherwise $\pint(t) = \nil$. 
  Using the induction hypothesis $\pint(s_i) \poel[\ell] \pint(t_i)$ and Lemma~\eref{l:poel:approx}{seq} we obtain
  \begin{equation*}
    \begin{array}{r@{\,}c@{\,}l}
      \pint(s) 
      & = & 
      \fn(\seq[m]{s}) \append \pint(s_{m+1}) \append \cdots \append \pint(s_i) \append \cdots \append \pint(s_{m+n})\\
      & \poel[\ell] & 
      \fn(\seq[m]{s}) \append \pint(s_{m+1}) \append \cdots \append \pint(t_i) \append \cdots \append \pint(s_{m+n})\\
      & = & \pint(t)
      \tkom
    \end{array}
  \end{equation*}
  as desired.
\end{proof}

\begin{proof}[of Theorem~\ref{t:sound}]
  Consider a derivation
  $$
  f(\pseq[k][l]{v}) = t_0 \rew[\RS] t_1 \rew[\RS] \cdots \rew[\RS] t_{m} \tkom
  $$
  for values $v_i$ ($i=1,\dots,k+l$)
  with respect to a \poecompatible\ TRS $\RS$. 

  Define $\ell \defsym \max\{\size{r} \mid l \to r \in \RS\}$. 
  Since $t_0 \in \Tn$, Lemma~\ref{l:tnderiv} shows that that $t_i \in \Tn$ for all $i = 1,\dots,m$. 
  As a consequence of Lemma~\ref{l:embed}, using 
  Lemma~\eref{l:poel:approx}{inc}, 
  we obtain
  $$
  \pint(t_0)
  \poel[\ell] \pint(t_{1})
  \cdots
  \poel[\ell] \pint(t_{m})
  \tpkt
  $$
  So in particular the length $m$ is bounded by 
  the length of $\poel[\ell]$ descending sequences starting 
  from $\lst{\fn(\vec{u})}$, i.e., 
  $$
  m \leqslant \Slow(\pint(t_0)) = \Slow(\fn(\seq[k]{v}))
  \tpkt
  $$
  Here the equality is given definition of $\pint$ and by Lemma~\ref{l:slowsum}. 
  The theorem follows thus from Lemma~\ref{l:poel}.
\end{proof}

%%% Local Variables: 
%%% mode: latex
%%% TeX-master: "paper"
%%% End: 

\section{Conclusion}

Adopting former works \cite{AEM11,AEM12}, we introduced a reduction
order, the Path Order for ETIME (\POESTAR), that is sound and complete for ETIME
computable functions.
The path order \POESTAR\ is a strictly intermediate order between the (small) 
path order for polytime (\POPSTARS) and exponential path
order (\EPOSTAR).

These orders differ only in constraints imposed on recursive definitions:
\POESTAR\ extends \POPSTARS\ by allowing nested recursive calls, 
as in the TRS $\RS_{\m{exp}}$; 
the order \EPOSTAR\ permits additionally recursion along lexicographic descending 
arguments, as in rule 
$\rlbl{5}$ of the TRS $\RS_{\m{fac}}$. 
Consequently, from our three examples only the TRS $\RS_{\m{add}}$ is compatible 
with \POPSTARS, whereas $\RS_{\m{add}}$ and $\RS_{\m{exp}}$ is compatible with \POESTAR\ 
and \EPOSTAR\ can even handle $\RS_{\m{fac}}$. 
%
% For example,
% \EPOSTAR\ can handle all of $\RS_{\m{add}}$, $\RS_{\m{exp}}$ and 
% $\RS_{\m{fac}}$,
% \POESTAR\ can handle $\RS_{\m{add}}$ and $\RS_{\m{exp}}$, but
% \POPSTARS\ can only handle $\RS_{\m{add}}$ among the three systems.

% Notably, the difference in definition lies only in Case~\ref{d:poe:ts}
% in Definition~\ref{d:poe}.
% In \POPSTARS, only element-wise comparison, i.e., product comparison, is
% allowed not only for normal arguments but also for safe arguments.
% On the other side, in \EPOSTAR, for normal arguments even a restrictive lexicographic comparison is allowed.
This contrast clarifies the relationship 
$\mathrm{P} \subseteq \mathrm{ETIME} \subseteq \mathrm{EXP}$
for the class $\mathrm{P}$ of polytime predicates and 
the class $\mathrm{EXP}$ of exponential-time ones.

%\bibliographystyle{abbrv}
%\bibliography{references}

\begin{thebibliography}{10}

\bibitem{AE09}
T.~Arai and N.~Eguchi.
\newblock {A New Function Algebra of EXPTIME Functions by Safe Nested
  Recursion}.
\newblock {\em ACM Transactions on Computational Logic}, 10(4), Article No. 24, 19 pages, 2009.

\bibitem{AEM11}
M.~Avanzini, N.~Eguchi, and G.~Moser.
\newblock {A Path Order for Rewrite Systems that Compute Exponential Time
  Functions}.
\newblock In {\em Proceedings of the 22nd International Conference on Rewriting
  Techniques and Applications (RTA 2011)}, volume~10 of {\em Leibniz
  International Proceedings in Informatics}, pages 123--138, 2011.

\bibitem{AEM12}
M.~Avanzini, N.~Eguchi, and G.~Moser.
\newblock {A New Order-theoretic Characterisation of the Polytime Computable
  Functions}.
\newblock In {\em Proceedings of the 10th Asian Symposium on Programming
  Languages and Systems (APLAS 2012)}, volume 7705 of {\em Lecture Notes in
  Computer Science}, pages 280--295, 2012.

\bibitem{AM10}
M.~Avanzini and G.~Moser.
\newblock {Closing the Gap Between Runtime Complexity and Polytime
  Computability}.
\newblock In {\em Proceedings of the 21st International Conference on Rewriting
  Techniques and Applications (RTA 2010)}, volume~6 of {\em Leibniz
  International Proceedings in Informatics}, pages 33--48, 2010.

\bibitem{baader}
F.~Baader and T.~Nipkow.
\newblock {\em Term Rewriting and All That}.
\newblock Cambredge University Press, 1998.

\bibitem{BC92}
S.~Bellantoni and S.~A. Cook.
\newblock {A New Recursion-theoretic Characterization of the Polytime
  Functions}.
\newblock {\em Computational Complexity}, 2(2):97--110, 1992.

\bibitem{clote97}
P.~Clote.
\newblock {A Safe Recursion Scheme for Exponential Time}.
\newblock In {\em Proceedings of the 4th International Symposium, Logical
  Foundation of Computer Science}, volume 1234 of {\em Lecture Notes in
  Computer Science}, pages 44--52, 1997.

\bibitem{LM09}
U.~Dal Lago and S.~Martini.
\newblock On {C}onstructor {R}ewrite {S}ystems and the {L}ambda-{C}alculus.
\newblock In {\em Proceedings of the 36th International Colloquium on Automata,
  Languages and Programming (ICALP 2009)}, volume 5556 of {\em Lecture Notes in
  Computer Science}, pages 163--174, 2009.

\bibitem{Leivant95}
D.~Leivant.
\newblock {Ramified Recurrence and Computational Complexity I: Word Recurrence
  and Poly-time}.
\newblock In P.~Clote and J.~B. Remmel, editors, {\em Feasible Mathematics II,
  Progress in Computer Science and Applied Logic}, volume~13, pages 320--343.
  {Birkh\"{a}user Boston}, 1995.

\bibitem{Monien77}
B.~Monien.
\newblock {A Recursive and Grammatical Characterization of Exponential Time
  Languages}.
\newblock {\em Theoretical Computer Science}, 3:61--74, 1977.

\end{thebibliography}

%\input{bib}

%\newpage
%\appendix

\end{document}